\documentclass[10pt,oneside,reqno]{amsart}
\usepackage{amsfonts,amssymb,
amscd,amsmath,latexsym,amsbsy,bm}
\numberwithin{equation}{section}
\theoremstyle{plain}

 \usepackage[usenames,dvipsnames]{xcolor}

 \usepackage[colorlinks,hyperindex, 
             bookmarks=true,bookmarksopen=true]{hyperref}
     \hypersetup
     {
         colorlinks=true,
         linkcolor=MidnightBlue,
         urlcolor=MidnightBlue,
         filecolor=black,
         citecolor=RedOrange,
         pdfstartview=FitV,
         pdftitle={},
         pdfauthor={Ilmar Gahramanov},
         pdfsubject={},
         pdfkeywords={},
         pdfpagemode=None,
         bookmarksopen=true
     }

\usepackage{latexsym}
\usepackage{amssymb,amsmath,amsfonts}
\usepackage{verbatim}
\usepackage{changebar}
\usepackage{setspace}
\usepackage{simplewick}
\usepackage{tocvsec2}
\usepackage{array}
\usepackage{cite}
\usepackage{mathtext}
\usepackage{stackrel}

 \usepackage{mathtools,amsmath,amsthm,amssymb,amsfonts}

 \numberwithin{equation}{section} 
 \theoremstyle{plain}
 \newtheorem{theo+}           {Theorem}      [section]
 \newtheorem{prop+}  [theo+]  {Proposition}
 \newtheorem{coro+}  [theo+]  {Corollary}
 \newtheorem{lemm+}  [theo+]  {Lemma}
 \newtheorem{defi+}  [theo+]  {Definition}
 \newtheorem{conj+}  [theo+]  {Conjecture}
 
 \theoremstyle{definition}
 \newtheorem{rema+}  [theo+]  {Remark}
 \newtheorem{prob+}  [theo+]  {Problem}
 \newtheorem{exam+}  [theo+]  {Example}
 
 \newenvironment{theorem}{\begin{theo+}}{\end{theo+}}
 \newenvironment{proposition}{\begin{prop+}}{\end{prop+}}

 \newcommand{\ti}{\mathrm i}

 \newcommand{\qhyp}[5]{\,\mbox{}_{#1}\phi_{#2}\!\left[
   \genfrac{}{}{0pt}{}{#3}{#4};#5\right]}
 \newcommand{\bhyp}[5]{\,\mbox{}_{#1}\psi_{#2}\!\left[
   \genfrac{}{}{0pt}{}{#3}{#4};#5\right]}   
 \newcommand{\whyp}[3]{\,\mbox{}_{#1}W_{#2}\!\left(#3\right)}

\begin{document}

\baselineskip 18pt
\larger[2]
\title[Basic hypergeometry of supersymmetric dualities]{Basic hypergeometry of supersymmetric dualities} 
\author{Ilmar Gahramanov}
\address{Max Planck Institute for Gravitational Physics (Albert Einstein Institute), Am M\"{u}hlenberg 1, D-14476 Potsdam, Germany $\&$ Institut f\"{u}r Physik und IRIS Adlershof, Humboldt-Universit\"{a}t zu Berlin,
Zum Grossen Windkanal 6, D12489 Berlin, Germany $\&$ Institute of Radiation Problems ANAS, B.Vahabzade 9, AZ1143 Baku, Azerbaijan $\&$ School of Engineering and Applied Science, Khazar University, Mehseti St. 41, AZ1096, Baku, Azerbaijan}
\email{ilmar.gahramanov@aei.mpg.de}
\author{Hjalmar Rosengren}
\address
{Department of Mathematical Sciences
\\ Chalmers University of Technology and University of Gothenburg\\SE-412~96 G\"oteborg, Sweden}
\email{hjalmar@chalmers.se}
\urladdr{http://www.math.chalmers.se/{\textasciitilde}hjalmar}

\keywords{Basic hypergeometric function, $q$-hypergeometric function, $q$-hypergeometric integral, superconformal index, supersymmetric duality, mirror symmetry, Seiberg duality, integral pentagon identity}
\subjclass{Primary 81T60, 33D60; Secondary 33E20, 33D90}

\begin{abstract}

We introduce several new identities combining  basic hypergeometric sums and integrals. Such identities appear in the context of superconformal index computations for three-dimensional supersymmetric dual theories. We give both analytic proofs and physical interpretations of the presented  identities. 

\end{abstract}

\maketitle        


\section{Introduction}

Recently, there has been renewed interest in basic hypergeometric integrals because of their connection with various branches of mathematical physics, such as supersymmetric field theory, 3-manifold invariants and integrable systems.
The purpose of this paper is to state and prove new basic hypergeometric integral identities and give their physical interpretations in terms of superconformal indices. 

There is an interesting connection between partition functions of supersymmetric gauge theories on different curved manifolds and certain classes of hypergeometric functions. The first observation of this relation was made by Dolan and Osborn \cite{Dolan:2008qi}. They found that the  superconformal index  of four-dimensional $\mathcal N=1$ supersymmetric gauge theory can be written via elliptic hypergeometric integrals. Similarly,   three-dimensional superconformal indices can be expressed in terms of basic hypergeometric integrals (see e.g. \cite{Krattenthaler:2011da,Kapustin:2011jm,Gahramanov:2013rda,Gahramanov:2014ona}).

The superconformal index for a three-dimensional ${\mathcal N}=2$ supersymmetric field theory is defined as 
\begin{equation}
\text{Tr} \left[ (-1)^\text{F} e^{-\beta\{Q, Q^\dagger \}} q^{\frac12 (\Delta+j_3)}\prod_i
t_i^{F_i} \right] \;,
\end{equation}
where the trace is taken over the Hilbert space of the theory,   $Q$ and ${Q}^\dagger$ are supercharges, $\Delta$, $j_3$ are Cartan elements of the superconformal group and the fugacities $t_i$ are associated with the flavor symmetry group. 

Studying the relation between basic hypergeometric integrals and  superconformal indices is an important field of research from different points of view (see e.g. \cite{Gahramanov:2015tta}). Non-trivial mathematical identities for superconformal indices provide a very powerful tool to check known supersymmetric dualities and to establish  new ones. Such identities are also important for  better understanding the structure of the moduli of three-dimensional supersymmetric theories and supersymmetric dualities. On the other hand, there is an interesting relationship between three-dimensional ${\mathcal N}=2$ supersymmetric gauge theories and geometry of triangulated 3-manifolds. The independence of a certain topological invariant of 3-manifolds on the choice of triangulation corresponds to equality of superconformal indices of three-dimensional ${\mathcal N}=2$ supersymmetric dual theories. 

Besides their appearance in  supersymmetric field theory, basic hypergeometric integrals discussed in this paper recently appeared in the theory of exactly solvable two-dimensional statistical models \cite{Gahramanov:2015cva,Kels:2015bda}. 

In this paper we extend the results of our previous work \cite{Gahramanov:2013rda,Gahramanov:2014ona} on superconformal indices to a number of three-dimensional dualities. We provide explicit expressions  for the generalized superconformal indices of some three-dimensional ${\mathcal N}=2$ supersymmetric electrodynamics and quantum chromodynamics in terms of basic hypergeometric integrals.

We will only consider confining theories, which means that the duality leads to a closed form evaluation of a sum of  integrals (rather than a transformation between two such expressions). As an example, one of the resulting identities is
\begin{multline}\label{wp}\sum_{m=-\infty}^\infty\oint \prod_{j=1}^6\frac{(q^{1+m/2}/a_jz,q^{1-m/2}z/a_j;q)_\infty}{(q^{N_j+m/2}a_jz,q^{N_j-m/2}a_j/z;q)_\infty}\frac{(1-q^mz^2)(1-q^mz^{-2})}{q^mz^{6m}}\frac{dz}{2\pi\ti z}\\
=\frac{2}{\prod_{j=1}^6q^{\binom {N_j}2}a_j^{N_j}}\prod_{1\leq j<k\leq 6}\frac{(q/a_ja_k;q)_\infty}{(a_ja_kq^{N_j+N_k};q)_\infty} \;,
\end{multline}
where $|q|<1$, the parameters $a_j$ are generic  and $N_j$ are integers, subject to
 the balancing conditions $\prod_{i=1}^6 a_i=q$ and $N_1+\dots+N_6=0$.
Here, we use the standard notation
$$(a;q)_\infty=\prod_{j=0}^\infty(1-aq^j)\;, $$
$$(a_1,\dots,a_m;q)_\infty=(a_1;q)_\infty\dotsm(a_m;q)_\infty$$
and the integration is over a positively oriented contour separating sequences of poles going to infinity from sequences going to zero.

The organization of the paper is as follows.

\begin{itemize}
\item In Section \ref{sis} we outline the superconformal index technique for three-dimensional ${\mathcal N}=2$ supersymmetric gauge theories.

\item In Section \ref{iis} we discuss supersymmetric dualities  and present  explicit expressions of superconformal indices for certain supersymmetric dual theories in terms of basic hypergeometric integrals. 
We present four examples, each leading to an integral evaluation similar
to \eqref{wp}. Some of these evaluations generalize identities
previously obtained in \cite{Krattenthaler:2011da, Kapustin:2011jm, Gahramanov:2013rda}.

\item 
 In Section \ref{mps} we give mathematical proofs of the four integral evaluations that were derived using non-rigorous methods in Section \ref{iis}.
 This gives a consistency check of the corresponding supersymmetric dualities. 

\item We review the basic aspects of three-dimensional $\mathcal N = 2$ supersymmetric gauge theories with focus on the necessary elements for the superconformal index computations and give some details of index computation in Appendices.
\end{itemize}

\section{3d superconformal index}
\label{sis}

In this section, we recall basic facts related to the superconformal index technique. The presentation closely follows that in \cite{Imamura:2011su,Krattenthaler:2011da,Kapustin:2011jm}.

The concept of the superconformal index was first introduced for four-di\-men\-sio\-nal theories in \cite{Romelsberger:2005eg,Kinney:2005ej} and later extended to other dimensions. The superconformal index of three-dimensional ${\mathcal N}=2$ superconformal field theory is a twisted partition
function defined on $S^2 \times S^1$ as follows \cite{Bhattacharya:2008bja,Kim:2009wb,Imamura:2011su}
\begin{equation}
{I}(q,\{ t_i \})=\text{Tr} \left[ (-1)^\text{F} e^{-\beta\{Q, Q^\dagger \}} q^{\frac12 (\Delta+j_3)}\prod_i
t_i^{F_i} \right] \;,
\end{equation}
where 
\begin{itemize}

\item the trace is taken over the Hilbert space of the theory on $S^2$,

\item $\text{F}$ plays the role of the fermion number which takes value zero on bosons and one on fermions. In presence of monopoles one needs to refine this number by shifting it by $e \times m$, where $e$ and $m$ are electric charge and magnetic monopole charge, respectively. See \cite{Dimofte:2011py,Aharony:2013dha} for a discussion of this issue.

\item $\Delta$ is the energy (or conformal dimension via radial quantization), $j_3$ is the third component of the angular momentum on $S^2$,

\item $F_i$ is the charge of global symmetry with fugacity $t_i$,  

\item $Q$ is a certain supersymmetric charge in three-dimensional ${\mathcal N}=2$ superconformal algebra with quantum numbers $\Delta=\frac12$ and $j_3=-\frac12$ and $R$-charge $R=1$. The supercharges $Q^\dagger = S$ and $Q$ satisfy the  anti-commutation relation\footnote{ The full algebra
can be found in many places, see e.g., \cite{Dolan:2008vc}.}
\begin{equation}\label{qs}
\frac12 \{Q, S\}=\Delta-R-j_3 \;.
\end{equation}
\end{itemize}

Only BPS states with $\Delta-R-j_3 =0 $ contribute to the superconformal index.
Consequently,  the index is $\beta$-independent but depends non-trivially on the fugacities $t_i$ and $q$. The superconformal index counts the number of BPS states weighted by their quantum numbers.

The superconformal index can be evaluated by a path integral on $S^2 \times S^1$  via the localization technique \cite{Pestun:2007rz}, leading to the  matrix integral \cite{Kim:2009wb,Imamura:2011su} 
\begin{align} \nonumber
I(q,\{t_i\})& = \sum_{m\in \mathbb{Z}^{\text{rank}\,G}} \int \frac{1}{|W_m|}
e^{-S^{(0)}_{CS}}e^{\ti b_0} q^{\frac12 \epsilon_0} \prod_{j=1}^{\text{rank}\, F} t_{j}^{q_{0j}} \\ \label{theindex}
& \qquad \qquad \times \exp\left[\sum^\infty_{n=1}\frac{1}{n} \text{ind}(z_i^{n},
t^n, q^n; m)\right] \; \prod_{i=1}^{\text{rank}\,G} \frac{dz_i}{2 \pi \ti z_i} \;.
\end{align}
The sum in the formula is to be understood as follows. It is a sum over magnetic fluxes $m=(m_1, \ldots, m_{\text{rank}\,G})$ on the two-sphere with
\begin{equation}
m_i \ = \ \frac{1}{2 \pi} \int_{S^2} F_i \;,
\end{equation}
where $m_i$ parametrizes the GNO charge of the monopole configuration\footnote{The operators creating magnetic fluxes are not completely understood yet, for details, see e.g. \cite{Kim:2009wb}.}, in the examples we consider it runs over the integers. The prefactor $|W_m|= \prod_{i=1}^k ( \text{rank}\,G_i)!$ is the order of the Weyl group of $G$ which is ``broken'' by the monopoles into the product $G_1 \times G_2 \times \dots\times G_k$. For instance, in case of $U(N)$ gauge group $|W_m|=\prod N_k!$.

The term 
\begin{align} \nonumber
S^{(0)}_{CS} & \ = \ \frac{\ti k}{4 \pi} \int \text{tr}_{CS} (A^{(0)} dA^{(0)}
-\frac{2 \ti} {3} A^{(0)} A^{(0)} A^{(0)}) \\ \label{CScont}
& \ = \ 2 \ti \text{tr}_{CS} (g m)\;,
\end{align}
is the contribution of the Chern--Simons term if the action contains such term and 
\begin{equation} \label{oneloopcont}
b_0 \ = \ -\frac{1}{2} \; \sum_\Phi\sum_{\rho\in R_\Phi}|\rho(m)|\rho(g)
\end{equation}
is the 1-loop correction to the Chern--Simons term. The $\text{tr}_{CS}$ stands for the trace containing the Chern--Simons levels, $k$ is the Chern--Simons level and $\sum_\Phi$ and $\sum_{\rho\in R_\Phi}$ are sums over all chiral multiplets and all weights of the representation $R_\Phi$, respectively. We give the contribution (\ref{CScont}) for completeness; in all our examples we will consider theories without the Chern--Simons term\footnote{Note that even in this case the term $b_0$ is not absent since the gauge fields generate the one-loop correction to the Chern-Simons term, see Appendixe B.}.

The term $q_{0j}$ in (\ref{theindex}) is the zero-point contribution to the energy,
\begin{equation}
q_{0j}(m) = -\frac{1}{2} \sum_\Phi \sum_{\rho\in R_\Phi} |\rho(m)| f_j (\Phi) \;.
\end{equation}
In addition, there is the contribution from the Casimir energy of the vacuum state on the two-sphere with magnetic flux $m$,
\begin{equation}
\epsilon_0(m) = \frac{1}{2} \sum_\Phi (1-\Delta_\Phi) \sum_{\rho\in
R_\Phi} |\rho(m)|
- \frac{1}{2} \sum_{\alpha \in G} |\alpha(m)| \;,
\end{equation}
where $\sum_{\alpha\in G}$ represents summation over all roots of $G$, $\Delta_\Phi$ is the superconformal $R$-charge of the chiral multiplet $\Phi$. 

One can calculate the single letter index
\begin{align}
& \text{ind}(z_j=e^{ig_j},t_j,q;m_j)  = -\sum_{\alpha\in G} e^{\ti\alpha(g)} q^{\frac12|\alpha(m)|}\\ \nonumber
& \qquad + \sum_\Phi \sum_{\rho\in R_\Phi} \left[
e^{\ti\rho(g)}  \prod_{j} t_j^{f_j}
\frac{q^{\frac12 |\rho(m)|+\frac12 \Delta_\Phi}}{1-q}  -  e^{-\ti\rho(g)} 
\prod_{j} t_j^{-f_j} \frac{q^{\frac12 |\rho(m)|+1-\frac12 \Delta_\Phi}}{1-q}
\right] \;.
\end{align}
Here, the first term is the contribution of the vector multiplets and the second line is the contribution of matter multiplets, labeled by $\Phi$, where $j$ runs over the rank of the flavor symmetry group. Given the single letter index it is a combinatorial problem \cite{Benvenuti:2006qr,Feng:2007ur} to compute the full multi-letter index. The result is given by the so-called ``plethystic'' exponential
\begin{equation}
\exp \bigg( \sum_{n=1}^\infty \frac{1}{n} \text{ind} ( z^n, t^n, q^n; m) \bigg) \;.
\end{equation}
For instance, let us consider the ${\mathcal N}=2$ theory with $U(N)$ gauge group. Then, the chiral multiplet $\Phi$ with $R$-charge $r$ in the fundamental representation of the gauge group contributes to the single-letter index as
\begin{equation}
\sum_{i=1}^{N} \left[ z_i t^{f(\Phi)} \frac{q^{\frac{r}{2}+\frac{|m_i|}{2}}}{1-q}-z_i^{-1}t^{-f(\Phi)} \frac{q^{1-\frac{r}{2}+\frac{|m_i|}{2}}}{1-q} \right] \;.
\end{equation}
After the ``plethystic'' exponential one obtains the contribution of the chiral multiplet to the index
\begin{equation} \label{contmul}
\prod_{i=1}^{N} \frac{(q^{1-\frac{r}{2}+\frac{|m_i|}{2}} t^{-f(\Phi)}z_i^{-1}; q)_{\infty}}
{(q^{\frac{r}{2}+\frac{|m_i|}{2}} t^{f(\Phi)}z_i; q)_{\infty}} \;.
\end{equation}
Similarly the contribution of the vector multiplet to the single-letter index is
\begin{equation}
-\sum_{i,j=1,\ldots, N,\, i\neq j} q^{\frac14 |m_i-m_j|} \frac{z_i}{z_j} \;,
\end{equation}
and the multi-letter index gets the  form
\begin{equation}
q^{-\sum_{1\leq i<j\leq N} \frac{|m_i-m_j|}{2}} \prod_{i,j=1,\ldots, N,\, i\neq j}
\left(1-\frac{z_i}{z_j} q^{\frac{|m_i-m_j|}{2}}\right) \;.
\end{equation}

Our main interest is the so-called generalized superconformal index which includes integer parameters corresponding to global symmetries. In \cite{Kapustin:2011jm} Kapustin and Willett pointed out that one can generalize the superconformal index of three-dimensional supersymmetric gauge theory by considering the theory in a non-trivial background gauge field coupled to the global symmetries of the theory. As a result the superconformal index includes new discrete parameters for global symmetries; we do not sum over these parameters. In case of the generalized superconformal index the contribution (\ref{contmul}) has the  form 
\begin{equation}
\prod_{i=1}^{N} \frac{(q^{1-\frac{r}{2}+\frac{|m_i+ f(\Phi) n_\Phi|}{2}}
t^{-f(\Phi)}z_i^{-1}; q)_{\infty}}{(q^{\frac{r}{2}+\frac{|m_i+f(\Phi) n_\Phi|}{2}} t^{f(\Phi)}z_i; q)_{\infty}} \;,
\end{equation}
where the parameters $n_\Phi$ are new discrete variables. It is convenient to express the index as a product of contributions from chiral and vector multiplets
\begin{align} \nonumber
I(q, \{t_a\}, \{n_a\}) = \sum_{m_1,\ldots , m_{\text{rank}(G)}}\frac{1}{|W_m|} \oint \prod_{j=1}^{\text{rank}\,G} \frac{dz_j}{2 \pi \ti z_j} Z_{\text{gauge}} (z_j,m_j; q) \\  \label{gindex}
\times \prod_{\Phi} Z_{\Phi}(z_j, m_j; t_a, n_a; q) \;,
\end{align}
where 
\begin{equation}
Z_{\text{gauge}} (z_j,m_j; q) \ = \ \prod_{\alpha \in \text{ad}(G)} q^{-\frac12 |\alpha(m)|} \left( 1- e^{\alpha(g)} q^{\frac{|\alpha(m)|}{2}} \right)
\end{equation}
and 
\begin{align} \nonumber
Z_\Phi \ = \ \prod_{\rho \in R_\Phi} \left( q^{\frac{1-r_\phi}{2}} \prod_{j} e^{-\ti \rho(g)} t(\Phi)^{-f(\Phi)} \right)^{\frac12 |\rho(m)+f(\Phi) n(\Phi)|} \\
\times \frac{(e^{-\ti\rho(g)} t(\Phi)^{-f(\Phi)} q^{\frac12 |\rho(m)+f(\Phi) n(\Phi)|+\frac{1-r_\Phi}{2}}; q)_\infty}{(e^{\ti\rho(g)} t(\Phi)^{f(\Phi)} q^{\frac12 |\rho(m)+f(\Phi) n(\Phi)|+\frac{r_\Phi}{2}}; q)_\infty} \;.
\end{align}
Here $\text{ad}(G)$ stands for the adjoint representation of the gauge group $G$. Note that we do not write the contribution of the Chern--Simons term in (\ref{gindex}), since as we mentioned before we consider theories without this term. 

It is worth to mention that the three-dimensional superconformal index can be constructed from the so-called holomorphic blocks \cite{Beem:2012mb} due to its factorization property \cite{Krattenthaler:2011da,Pasquetti:2011fj,Taki:2013opa,Nieri:2013yra,Hwang:2012jh,Hwang:2015wna}, i.e. the superconformal index can be expressed in terms of two identical $3d$ holomorphic blocks ${\mathcal B}(x;q)$
as\footnote{Geometrically it means that the index can be obtained by gluing two solid tori. In this context $B_c(x;q)$ are partition functions on solid tori.}
\begin{equation}
\sum_c {\mathcal B}_c(x;q) {\mathcal B}_c(\tilde{x};\tilde{q}) \;.
\end{equation}
It is possible to obtain the factorized superconformal index directly from the localization technique via the so-called Higgs branch localization \cite{Fujitsuka:2013fga,Benini:2013yva}.

\section{Integral identities from 3d dualities}
\label{iis}

In \cite{Seiberg:1994pq} Seiberg found that there exist pairs of different four-dimensional ${\mathcal N}=1$ supersymmetric gauge theories which describe the same physics in the infrared limit. This is called supersymmetric duality. Since its proposal a large number of dualities in various dimensions have been found. 

In this section, we study three-dimensional ${\mathcal N} = 2$  supersymmetric dualities \cite{Intriligator:1996ex,Aharony:1997gp, Karch:1997ux, deBoer:1997ka} and demonstrate the matching of the superconformal index for dual theories. The superconformal index technique is one of the main tools for establishing and checking supersymmetric dualities.

In this work, we consider only confining theories, i.e.\ theories whose infrared limit can be described in terms of gauge invariant composites (mesons and baryons) and without dual quarks. There are definitely more confining supersymmetric theories in three dimensions (for recent discussions, see \cite{Csaki:2014cwa, Amariti:2015kha}). We restrict our attention to samples of theories with $U(1)$ (supersymmetric quantum electrodynamics) and $SU(2)$ (supersymmetric quantum chromodynamics) gauge symmetry. We also limit ourselves to the cases of vanishing Chern-Simons term; however, one can add such a term to the action of the theories considered in the paper.

Note that similar results for ${\mathcal N}=1$ supersymmetric gauge theories in four dimensions were intensively studied in \cite{Dolan:2008qi,Spiridonov:2009za,Spiridonov:2011hf}. All $3d$ dualities considered in the next section can be obtained via dimensional reduction from $4d$ dualities. However, obtaining the right duality in three dimensions is more tricky (for details see \cite{Aharony:2013dha,Niarchos:2012ah}). The main issue is that the reduction procedure and renormalization group flow from ultraviolet to infrared do not commute with each other. This happens because of an anomalous $U(1)$ symmetry in $4d$, which  one needs to break  in $3d$ theory. This can be done by adding a monopole operator to the $3d$ Lagrangians. To be more precise we need to add the effective superpotential $W=\eta X$ to the Lagrangian of electric theory and $W=\tilde{\eta} \tilde{X}$ to the magnetic theory (dual theory), where $X$ is a monopole operator and $\eta$ is the $4d$ instanton factor.

In our examples we give only the necessary input to compute the superconformal index and do not discuss other aspects of dual theories. As for many other dualities in physics, systematic proofs of supersymmetric dualities are absent and the superconformal index computations do not constitute a proof of the duality. There are other important arguments for three-dimensional supersymmetric dualities, i.e. study of superpotentials for interactions among chiral multiplets \cite{Aharony:2013dha}, brane construction (see e.g, \cite{deBoer:1997ka, Amariti:2015yea}), contact terms (see e.g., \cite{Closset:2012vp,Amariti:2014lla}) and other powerful methods very much in the spirit of the superconformal index such as study of sphere partition functions \cite{Kapustin:2009kz,Jafferis:2010un}, ellipsoid partition functions \cite{Dolan:2011rp,Niarchos:2012ah,Gahramanov:gka}, lens partition functions \cite{Imamura:2012rq, GahKels}, etc.

The 't Hooft anomaly matching conditions which played a central role in checking Seiberg dualities for ${\mathcal N} = 1$ supersymmetric gauge theories become  useless in three dimensions since, unlike four-dimensional gauge theories, in three dimensions there are no chiral anomalies. In three dimensions it is possible to have a classical Chern-Simons term which breaks parity. One can then use the matching condition for the parity anomaly; however, conditions for discrete anomalies are weaker than those for continuous anomalies. 

In what follows, we omit the $R$-charges for chiral multiplets, since the superconformal indices of dual theories match for arbitrary assignment of the $R$-charge \cite{Imamura:2011su}. The correct $R$-charges for matter fields in the infrared fixed points can be obtained by the so-called $Z$-extremization procedure \cite{Jafferis:2010un}.

As a final remark, let us comment that the matching of superconformal indices for dual pairs were studied mainly by expanding in terms of fugacities \cite{Imamura:2011su, Kapustin:2011vz, Kim:2013cma, Park:2013wta} and only in a few works \cite{Krattenthaler:2011da,Kapustin:2011jm,Gahramanov:2013rda} authors give rigorous proofs of the index identities. 

Below we give explicit expressions of generalized superconformal indices for some theories.  Equality of indices for dual theories 
leads to integral evaluations, which will be proved rigorously  in Section \ref{mps}.

\medskip

{\textbf{Example 1.}}

\medskip

We first consider a \textbf{Theory A} and its low-energy description \textbf{Theory B} which can be described purely in terms of composite gauge singlets. 

\medskip

\begin{itemize}

\item \textbf{Theory A}: Supersymmetric Quantum Chromodynamics with $SU(2)$ gauge group and with $SU(6)$ flavor group, chiral multiplets in the fundamental representation of the gauge group and the flavor group, a vector multiplet in the adjoint representation of the gauge group. Note that in case of $SU(2)$ gauge theories the fundamental and antifundamental representations are equivalent, therefore we have $SU(6)$ flavor group rather than $SU(3)\times SU(3) \times U(1)$.

\medskip

\item \textbf{Theory B}: no gauge symmetry, fifteen chiral multiplets in the totally antisymmetric tensor representation of the flavor group.

\end{itemize}

\medskip

This duality was considered in \cite{Teschner:2012em} where the authors presented the sphere partition functions for dual theories. It is analogous to the four-dimensional duality for similar theories \cite{Dolan:2008qi} and can be obtained by dimensional reduction.

Using the group-theoretical data (see Appendix B) it is straightforward to compute explicitly the generalized superconformal indices, and due to the supersymmetric duality we find the  basic hypergeometric integral identity
 
\begin{align} \nonumber
& \sum_{m\in\mathbb{Z}}\oint \frac{dz}{4 \pi \ti z} q^{-|m|} (1-q^{|m|} z^2)(1-q^{|m|} z^{-2})  \; (-q^{\frac 12})^{\sum_{i=1}^6 (\frac{|n_i+m|}{2}+\frac{| n_i-m|}{2})} \\ \nonumber
& \qquad \times z^{-\sum_{i=1}^6 (\frac{|n_i+m|}{2}-\frac{|n_i-m|}{2})} 
 \prod_{j=1}^6 a_j^{-\frac{|n_j+m|}{2}-\frac{|n_j-m|}{2}}
\frac{(q^{1+\frac{|n_j+m|}{2} }/{a_jz},q^{1+\frac{|n_j-m|}{2} }{z}/{a_j};q)_\infty}
{(q^{\frac{|n_j+m|}{2} }a_jz,q^{\frac{|n_j-m|}{2} }{a_j}/{z};q)_\infty}
 \\  \label{iden1}
& \qquad \qquad =(- q^{\frac12})^{\sum_{1\leq j<k\leq 6} \frac{|n_j+n_k|}{2}}
\prod_{1\leq j<k\leq 6} (a_j a_k)^{-\frac{|n_j+n_k|}{2}} \frac{(q^{1+\frac{|n_j+n_k|}{2}}a_j^{-1}a_k^{-1};q)_\infty}
{(q^{\frac{|n_j+n_k|}{2}}a_ja_k;q)_\infty}
\end{align}
with the balancing conditions 
\begin{equation} \label{BC6}
\prod_{i=1}^6 a_i=q \;\; \text{and} \;\; \sum_{i=1}^6 n_i=0 \;.
\end{equation} 

This identity describes confinement without breaking of the ``chiral symmetry''. The left side of the expression (\ref{iden1}) includes the contributions of twelve chirals and a vector multiplet, while the right hand side contains the contribution of fifteen chirals. From the fact that all physical degrees of freedom of Theory B are gauge invariant there is no integration on the right hand side. 

It is worth mentioning that the duality considered in the example is a special case of the duality claimed in \cite{Aharony:2013dha}, where the theory A is the three-dimensional $SP(2N)$ SQCD with $2N_f$ fundamentals and theory B is the $SP(2N_f-2N-4)$ theory with $2N_f$ fundamentals. Such duality is qualitatively similar to $SU(N)$ duality with matter in the fundamental representation of the gauge group. In case of $N=2$ one can consider the theory A as $SU(2)$ gauge theory since $SP(2) \simeq SU(2)$.

Note that the balancing conditions are imposed by the effective superpotential and the theories described above are dual only in the presence of certain superpotentials. We refer the interested reader to \cite{Aharony:2013dha} for more details related to the study of superpotentials for three-dimensional dualities.

In (\ref{iden1}) we used the absolute values of monopole charges as in the definition of the superconformal index. 
It is possible to eliminate all absolute values using the elementary
identity  \cite{Dimofte:2011py}

\begin{equation} \label{absolute}
\frac{(q^{1+|m|/2}/z;q)_\infty}{(q^{|m|/2}z;q)_\infty}=(-q^{-\frac 12}z)^{\frac{|m|-m}2}\frac{(q^{1+m/2}/z;q)_\infty}{(q^{m/2}z;q)_\infty}.
\end{equation}

After such simplification, \eqref{iden1} takes the form

\begin{multline*}\sum_{m=\in\mathbb Z}^\infty\oint \prod_{j=1}^6\frac{(q^{1+(m+n_j)/2}/a_jz,q^{1+(n_j-m)/2}z/a_j;q)_\infty}{(q^{(n_j+m)/2}a_jz,q^{(n_j-m)/2}a_j/z;q)_\infty}\frac{(1-q^mz^2)(1-q^mz^{-2})}{q^mz^{6m}}\frac{dz}{4\pi\ti z}\\
=\frac{1}{\prod_{j=1}^6a_j^{n_j}}\prod_{1\leq j<k\leq 6}\frac{(q^{1+(n_j+n_k)/2}/a_ja_k;q)_\infty}{(q^{(n_j+n_k)/2}a_ja_k;q)_\infty} \;.\end{multline*}

After replacing $(a_j,n_j)\mapsto(a_jq^{N_j/2},N_j)$, this is \eqref{wp}.
We give a rigorous mathematical proof of this identity in
Theorem \ref{awt}.

The most intriguing physical interpretation of the formula (\ref{iden1}) stems from the role it plays as a star-triangle relation \cite{Kels:2015bda, Gahramanov:2015cva} for a certain two-dimensional statistical model. 

The integral identity (\ref{iden1}) can be obtained by reduction \cite{Yamazaki:2013fva, Benini:2011nc, Kels:2015bda} from the similar  identity for four-dimensional lens indices. In \cite{Kels:2015bda} such reduction was made in the context of integrable statistical models.  

The $q \rightarrow 1$ limit of (\ref{iden1}) was discussed in \cite{Kels:2015bda}. This limit also has an interpretation in terms of exactly solvable statistical models \cite{Kels:2013ola}.  From the viewpoint of supersymmetric dualities such reduction \cite{Benini:2012ui,Benini:2011nc} gives the equality of the sphere partition functions of dual two-dimensional ${\mathcal N}=(2,2)$ supersymmetric gauge theories.

\medskip

{\textbf{Example 2.}}

\medskip

Our second example is again a supersymmetric quantum chromodynamics with a weakly coupled magnetic dual.

\medskip

\begin{itemize}

\item \textbf{Theory A}: Supersymmetric Quantum Chromodynamics with $SU(2)$ gauge group and four flavors, chiral multiplets in the fundamental representation of the gauge group and the flavor group, the vector multiplet in the adjoint representation of the gauge group.

\medskip 

\item \textbf{Theory B}: no gauge degrees of freedom, with six mesons and a singlet chiral field.

\medskip

\end{itemize}

\medskip

According to the supersymmetric duality we have the following integral
identity for the generalized superconformal indices:

\begin{align} \nonumber
& \sum_{m\in\mathbb{Z}}\oint \frac{dz}{4 \pi \ti z}\,q^{-|m|} (1-q^{|m|} z^2)(1-q^{|m|} z^{-2})  \; ({ -}q^{\frac 12})^{\sum_{i=1}^4 (\frac{|n_i+m|}{2}+\frac{|n_i-m|}{2}-{ n_i})} \\ \nonumber
& \qquad \times  z^{-\sum_{i=1}^4 (\frac{|n_i+m|}{2}-\frac{|n_i-m|}{2})} 
\prod_{j=1}^4 a_j^{-\frac{|n_j+m|}{2} { -}\frac{|n_j-m|}{2}{ +n_j}}
\frac{(q^{1+\frac{|n_j+m|}{2} }/{a_jz},q^{1+\frac{|n_j-m|}{2} }{z}/{a_j};q)_\infty}
{(q^{\frac{|n_j+m|}{2} }a_jz,q^{\frac{|n_j-m|}{2} }{a_j}/{z};q)_\infty}
 \\ \nonumber
& \quad = (-q^{\frac 12})^{\sum_{1\leq j<k\leq 4} \frac{|n_j+n_k|}{2} -\sum_{i=1}^4 n_i-\frac {| \sum_{i=1}^4 n_i|}{2}} (a_1 a_2 a_3 a_4)^{\frac{ |\sum_{i=1}^4  n_i| -\sum_{i=1}^4 n_i}{2}}  \\ \label{iden2}
& \times \frac{( q^{\frac{|\sum_{i=1}^4 n_i|}{2}}a_1 a_2 a_3 a_4)_{\infty}}{(q^{1+ \frac{|\sum_{i=1}^4 n_i|}{2}}/a_1 a_2 a_3 a_4)_{\infty}} \prod_{1\leq j<k\leq 4} (a_j a_k)^{-\frac{|n_j+n_k| +(n_j+n_k)}{2}} \frac{(q^{1+\frac{|n_j+n_k|}{2}}/a_ja_k;q)_\infty}
{(q^{\frac{|n_j+n_k|}{2}}a_ja_k;q)_\infty}.
\end{align}

The ordinary index of the theory A was considered in \cite{Gahramanov:2013xsa} in the context of global symmetry enhancement. It was shown that the superconformal index of the theory has an extended $SO(10)$ flavor symmetry when coupled to $4d$ multiplets with specific boundary conditions.

Note that one can deform dual theories from Example 1 by adding mass terms for some of the quarks. After integrating out one flavor (massive modes) the theory with the remaining four flavors confines with ``chiral symmetry breaking'' \cite{Gahramanov:2015cva} if we keep a certain superpotential for the theory giving the balancing conditions similar to (\ref{BC6}). Here the theory A has no superpotential and therefore we obtain the duality ({\ref{iden2}}).

Eliminating the absolute values as before, \eqref{iden2} can be expressed as
\begin{align} \nonumber
& \sum_{m\in\mathbb{Z}}\oint \frac{dz}{4 \pi \ti z}\frac{(1-q^{m} z^2)(1-q^{m} z^{-2})}{q^m     z^{4m} }
\prod_{j=1}^4 
\frac{(q^{1+\frac{n_j+m}{2} }/{a_jz},q^{1+\frac{n_j-m}{2} }{z}/{a_j};q)_\infty}
{(q^{\frac{n_j+m}{2} }a_jz,q^{\frac{n_j-m}{2} }{a_j}/{z};q)_\infty}
  \\ \label{iden2b}
&\quad= \frac{( q^{\frac{\sum_{i=1}^4 n_i}{2}}a_1 a_2 a_3 a_4)_{\infty}}{(q^{1+ \frac{\sum_{i=1}^4 n_i}{2}}/a_1 a_2 a_3 a_4)_{\infty}} \prod_{1\leq j<k\leq 4}  \frac{(q^{1+\frac{n_j+n_k}{2}}/a_ja_k;q)_\infty}
{(q^{\frac{n_j+n_k}{2}}a_ja_k;q)_\infty}.
\end{align}

This can be recognized as a special case of Proposition \ref{wps}. More precisely, Proposition \ref{wps} states that \eqref{iden2b} holds
even with the integers $n_j$ replaced by generic complex parameters.

\medskip

In contrast to four dimensions, there exist supersymmetric dualities for abelian gauge theories in three dimensions. For details of such dualities see e.g.\ \cite{Strassler:1999hy}.  Below we consider two examples of such dualities.

\medskip

{\textbf{Example 3.}}

\medskip

\begin{itemize}

\item \textbf{Theory A}: $d=3$ ${\mathcal N}=2$ supersymmetric electrodynamics with $U(1)$ gauge symmetry and six chiral multiplets, half of them transforming in the fundamental representation of the gauge group and another half transforming in the anti-fundamental representation.

\medskip
 
\item \textbf{Theory B}: no gauge degrees of freedom, nine gauge invariant ``mesons'' transforming in the fundamental representation of the flavor group.

\end{itemize}

\medskip

Supersymmetric duality leads to the following identity for the generalized super\-conformal indices: 

\begin{align} \nonumber
& \sum_{m \in \mathbb Z}  \oint \frac{dz}{2\pi \ti z}  (-q^{\frac 12})^{ \sum_{i=1}^3 (\frac{|m_i+m|}{2}+\frac{|n_i-m|}{2}) } z^{- \sum_{i=1}^3 (\frac{|m_i+m|}{2}-\frac{|n_i-m|}{2})}\\ \nonumber
& \qquad \qquad \times  \prod_{i=1}^3  a_i^{-\frac{|m_i+m|}{2}} b_i^{-\frac{|n_i-m|}{2}} \frac{(q^{1+\frac{|m_i+m|}{2}}/a_i z, q^{1+\frac{|n_i-m|}{2}} z/b_i ;q)_{\infty}}{(q^{\frac{|m_i+m|}{2}} a_i z ,q^{\frac{|n_i-m|}{2}} b_i/z ;q)_{\infty}} \\ \label{prepentagon}
& \qquad \qquad  = (-q^{\frac 12})^{ \sum_{i,j=1}^3\frac{|m_i+n_j|}{2}} \prod_{i,j=1}^{3} (a_ib_j)^{-\frac{|m_i+n_j|}{2}} \frac{(q^{1+\frac{|m_i+n_j|}{2}}/ a_i b_j;q)_{\infty}}{(q^{\frac{|m_i+n_j|}{2}} a_i b_j ;q)_{\infty}}   \;,
\end{align}
where the fugacities $a_i$ and $b_i$ stand for the flavor symmetry $SU(3)\times SU(3)$, $z$ is the fugacity for the $U(1)$ gauge group and the balancing conditions are 
\begin{align} \label{balcond}
\prod_{i=1}^3 a_i  = \prod_{i=1}^3  b_i =q^{\frac12} \;\; \text{and} \;\; \sum_{i=1}^3 n_i  =\sum_{i=1}^3 m_i=0 \;.
\end{align}

Eliminating the absolute values, this identity takes the form
\begin{multline*}\sum_{m=-\infty}^\infty\oint \prod_{i=1}^3\frac{(q^{1+(m+m_i)/2}/a_iz,q^{1+(n_i-m)/2}z/b_i;q)_\infty}{(q^{(m+m_i)/2}a_iz,q^{(n_i-m)/2}b_j/z;q)_\infty}\frac{1}{z^{3m}}\frac{dz}{2\pi\ti z}\\
=\frac{1}{\prod_{i=1}^3a_i^{m_i}b_i^{n_i}}\prod_{i,j=1}^3\frac{(q^{1+(m_i+n_j)/2}/a_ib_j;q)_\infty}{(q^{(m_i+n_j)/2}a_ib_j;q)_\infty}\;
. \end{multline*}

After the change of variables $z\mapsto -z$, one may check that this is equivalent to Theorem \ref{t} below. This identity was first 
proved in
 \cite{Gahramanov:2013rda} in the special case of
 ordinary superconformal indices\footnote{Note that the identity of sphere partition functions for this duality was presented in \cite{Kashaev:2012cz,Benvenuti:2016wet}.}, that is, $m_i\equiv n_i\equiv 0$.
The general case was presented without proof in \cite{Gahramanov:2014ona} .

The expression (\ref{prepentagon}) can be written as an integral pentagon identity. Following \cite{Gahramanov:2013rda}, we introduce the function

\begin{align} \nonumber
{\mathcal B}_m[a, n;b, m] & = (-q^{\frac 12})^{\frac{|n|}{2}+\frac{|m|}{2}-\frac{|n+m|}{2}} a^{-\frac{|n|}{2}} b^{-\frac{|m|}{2}} (ab)^{\frac{|n+m|}{2}} \\ 
& \quad \times \frac{(q^{1+\frac{|n|}{2}}a^{-1},q^{1+\frac{|m|}{2}}b^{-1},q^{\frac{|n+m|}{2}}ab;q)_\infty}{(q^{\frac{|n|}{2}}a,q^{\frac{|m|}{2}}b,q^{1+\frac{|n+m|}{2}}(ab)^{-1};q)_\infty}\; ,
\end{align}
and rewrite the equality (\ref{prepentagon}) in terms of this function. We obtain the following integral pentagon identity in terms of ${\mathcal B}$ functions:
\begin{align} \nonumber
 \sum_{m\in\mathbb Z} \oint \frac{d z}{2 \pi \ti z}   \prod_{i=1}^3 {\mathcal B}[ a_i z, n_i+m; b_i z^{-1}, m_i-m] \\ \label{pentagon}
 = {\mathcal B}[a_1 b_2, n_1+m_2; a_3 b_1; n_3+m_1] \; {\mathcal B}[a_2 b_1, n_2+m_1; a_3 b_2, n_3+m_2] \;,
\end{align}
with the balancing conditions (\ref{balcond}). 

The integral identity (\ref{pentagon}) is interesting from the following point of view. There is a recently proposed relation called $3d/3d$ correspondence between $3d$ ${\mathcal N}=2$ supersymmetric gauge theories and 3-manifolds \cite{Terashima:2011qi, Dimofte:2011ju} (see also \cite{Dimofte:2014ija,Dimofte:2011py} and earlier works \cite{Terashima:2011xe, Terashima:2012cx}) in similar spirit as the AGT correspondence \cite{Alday:2009aq}. This correspondence translates the ideal triangulation of the 3-manifold into mirror symmetry for three-dimensional supersymmetric theories. The independence of the corresponding 3-manifold invariant on the choice of triangulation corresponds to the equality of superconformal indices of mirror dual theories \cite{Dimofte:2011py}. In this context the identity (\ref{pentagon}) encodes a 3--2 Pachner move for 3-manifolds. 

\medskip

{\textbf{Example 4.}}

\medskip

Let us consider another example of abelian duality, namely the well-known XYZ/SQED mirror symmetry \cite{Aharony:1997bx,deBoer:1997ka,Intriligator:1996ex}.

\medskip

\begin{itemize}

\item \textbf{Theory A:} ${\mathcal N}=2$ supersymmetric quantum electrodynamics,  with a single $U(1)$ vector multiplet and two chiral multiplets charged oppositely under the gauge group.

\medskip

\item \textbf{Theory B:} free Wess--Zumino theory with three chiral multiplets. This theory is often is called the XYZ model in the literature.

\medskip

\end{itemize} 

In this example we wish to turn on the contribution to the generalized superconformal index of the topological symmetry $U(1)_J$, which is not explicit in the Lagrangian. This hidden symmetry is generated by the current
\begin{equation}
J^\mu \ = \  \varepsilon^{\mu \nu \rho} F_{\nu \rho} \;.
\end{equation}
The current $J^{\mu}$ is topologically conserved\footnote{The corresponding charge is carried by the Abrikosov-Nielsen-Olesen vortices in the Higgs branch of $\mathcal N=2$ theory.} due to the Bianchi identity. 

In this case we have a special duality called mirror symmetry which exchanges the Coulomb branch of a theory with the Higgs branch of its mirror dual and vice versa. The duality implies the identity
\begin{align} \nonumber
& \sum_{s \in \mathbb{Z}} \oint \frac{dz}{2 \pi \ti z} (-1)^{s+m+\frac{|s+ m|}{2}+\frac{|s - m|}{2}} z^{n} w^{s} (q^{\frac14} z \alpha^{-1})^{\frac{|s - m|}{2}} (q^{\frac14} z^{- 1} \alpha^{-1})^{\frac{|s+ m|}{2}} \\
\nonumber&\qquad\qquad \quad\times\frac{(z^{\pm } \alpha^{-1} q^{\frac{|s \mp m|}{2}+\frac34};q)_\infty}{(z^{\pm } \alpha q^{\frac{|s \pm m|}{2}+\frac14};q)_\infty} \\ \nonumber
& \qquad \qquad = (-1)^{n+m+\frac{|n+ m|}{2}+\frac{|n- m|}{2}} (q^{\frac14} \alpha w)^{\frac{|m - n|}{2}} (q^{\frac14} \alpha w^{- 1})^{\frac{|m+ n|}{2}} \alpha^{-2|m|} \\ \label{mirror2}
& \qquad \qquad \quad  \times \frac{(\alpha w^{\pm } q^{\frac{|m \pm n|}{2}+\frac34},\alpha^{-2} q^{|m|+\frac12};q)_\infty}{(\alpha^{-1} w^{\pm } q^{\frac{|m \mp n|}{2}+\frac14},\alpha^2 q^{|{m}|+\frac12};q)_\infty} \;,
\end{align}
where the fugacity $\alpha$ and the monopole charge $m$ denote the parameters for the axial $U(1)_A$ symmetry, $\omega$ and $n$ denote the parameters for the topological $U(1)_J$ symmetry and the discrete parameter $s$ stands for the magnetic charge corresponding to the $U(1)$ gauge group. 
The factors containing $\pm$ should be interpreted as the product over both choices; for instance,
$$(z^{\pm } \alpha^{-1} q^{\frac{|s \mp m|}{2}+\frac34};q)_\infty=(z\alpha^{-1} q^{\frac{|s- m|}{2}+\frac34};q)_\infty(z^{-1} \alpha^{-1} q^{\frac{|s+ m|}{2}+\frac34};q)_\infty. $$

Here, we explicitly write the $R$-charges of chiral multiplets. Due to the  permutation symmetry of the superpotential $W=\tilde{q} S q$ for the theory B, where  $q,\tilde{q}$, $S$ are three chiral multiplets of the theory, one can fix the $R$-charges\footnote{ In the infrared limit the superpotential $W$ must have the R-charge $2$, then the $R$-charge of chiral multiplets of theory B must be $\frac23$. }.

The case $m=n=0$ of (\ref{mirror2}) was presented in \cite{Imamura:2011su, Krattenthaler:2011da} and proven in \cite{Krattenthaler:2011da}.
The general case was presented, with a slight mistake\footnote{{ We have an additional phase factor $(-1)^{s+m+\frac{|s + m|}{2}+\frac{|s - m|}{2}}$, which is due to the definition of the fermion number operator $F$ in the definition of the superconformal index \cite{Dimofte:2011py} (see also \cite{Aharony:2013dha,Hwang:2015wna}). In fact, in general the superconformal indices match for dual theories in presence of this corrected phase factor \cite{Aharony:2013dha}. Without the phase factor the identity presented by Kapustin and Willett \cite{Kapustin:2011jm} is incorrect. It is actually a good example where the naive choice of the fermion number as $2J_3$ does not work.}}, in \cite{Kapustin:2011jm}, where 
a proof was given for the special case $m=0$. In Section \ref{mps} we give an analytic proof of the general case. More precisely, eliminating the absolute values in \eqref{mirror2} gives
\begin{align} \nonumber
& \sum_{s \in \mathbb{Z}} \oint \frac{dz}{2 \pi \ti z} (-w)^s z^{n-s}
\frac{(z^{\pm } \alpha^{-1} q^{\frac{m\mp s}{2}+\frac34};q)_\infty}{(z^{\pm } \alpha q^{\frac{m\pm s}{2}+\frac14};q)_\infty} \\ \nonumber
& \qquad \qquad = (-w)^n\frac{(\alpha w^{\pm } q^{\frac{m \pm n}{2}+\frac34},\alpha^{-2} q^{m+\frac12};q)_\infty}{(\alpha^{-1} w^{\pm } q^{\frac{m \pm n}{2}+\frac14},\alpha^2 q^{{m}+\frac12};q)_\infty} \;,
\end{align}
which can be recognized as the special case $a=b=q^{\frac 14-\frac m2}\alpha$, $c=d=q^{\frac 14+\frac m2}\alpha$ of Proposition~\ref{ramp}.

The identity \eqref{mirror2} and related identities can also be written as pentagon identities.
In fact, introducing the tetrahedron index \cite{Dimofte:2011py, Dimofte:2011ju}
$$\mathcal I_q[m,z] = \frac{(q^{1-\frac m2}/z;q)_\infty}{(q^{-\frac m2}z;q)_\infty},$$
 it follows from Proposition \ref{ramp} that 
 \begin{multline*}\sum_{s\in\mathbb Z}\oint(-w)^sz^{N-s}\,\mathcal I_q[m-s;q^{1/4}\alpha z]\,
\mathcal I_q[n+s;q^{1/4}\beta/z]\,\frac{dz}{2\pi\ti z}\\
=(-w)^N\mathcal I_q[m+n;q^{\frac 12}\alpha\beta]\,
\mathcal I_q[n+N;q^{1/4}w/\beta]\,\mathcal I_q[m-N;q^{1/4}/\alpha w].
\end{multline*}
Special cases with $m=n=N=0$ and $m=n$ (corresponding to \eqref{mirror2}) were presented earlier in \cite{Gahramanov:2013rda},
   \cite{Gahramanov:2014ona}, respectively.

One can also consider this duality as a mirror symmetry between ${\mathcal N}=4$ super\-symmetric electrodynamics with a single flavor and its dual theory with a free hypermultiplet. Then we obtain instead of (\ref{mirror2}) the mathematically equivalent identity 
\begin{align} \nonumber
& \alpha^{2|m|} \frac{(\alpha^{2} q^{|m|+\frac12};q)_\infty}{(\alpha^{-2} q^{|{m}|+\frac12};q)_\infty} \sum_{s \in \mathbb{Z}} \oint \frac{dz}{2 \pi \ti z} (-1)^{s+m+\frac{|s+ m|}{2}+\frac{|s - m|}{2}} \\ \nonumber
& \qquad \qquad \qquad \times z^{n} w^{s} (q^{\frac14} z \alpha^{-1})^{\frac{|s - m|}{2}} (q^{\frac14} z^{- 1} \alpha^{-1})^{\frac{|s+ m|}{2}}  \frac{(z^{\pm } \alpha^{-1} q^{\frac{|s \mp m|}{2}+\frac34};q)_\infty}{(z^{\pm } \alpha q^{\frac{|s \pm m|}{2}+\frac14};q)_\infty} \\ 
& = (-1)^{n+m+\frac{|n+ m|}{2}+\frac{|n- m|}{2}} (q^{\frac14} \alpha w)^{\frac{|m - n|}{2}} (q^{\frac14} \alpha w^{- 1})^{\frac{|m+ n|}{2}}   \frac{(\alpha w^{\pm } q^{\frac{|m \pm n|}{2}+\frac34} ;q)_\infty}{(\alpha^{-1} w^{\pm } q^{\frac{|m \mp n|}{2}+\frac14} ; q)_\infty} \;.
\end{align}

\section{Mathematical proofs of identities}
\label{mps}

In this section we will use the standard notation of \cite{gr}. In particular,
the basic hypergeometric series is \cite[Ex.\ (1.2.22)]{gr}
\begin{equation}
\qhyp rs{a_1,\ldots,a_r}{b_1,\ldots, b_s}{q,z} = \sum_{j=0}^{\infty} \frac{(a_1;q)_j \ldots (a_r;q)_j}{(b_1;q)_j \ldots (b_s;q)_j} \left[(-1)^{j}q^{{\binom j 2}}\right]^{1+s-r} z^j \;,
\end{equation}
and the bilateral basic hypergeometric series is \cite[Ex.\ (5.1.1)]{gr}
\begin{equation}
\bhyp rs{a_1,\ldots,a_r}{b_1,\ldots, b_s}{q,z} = \sum_{j=0}^{\infty} \frac{(a_1;q)_j \ldots (a_r;q)_j}{(b_1;q)_j \ldots (b_s;q)_j} (-1)^{(s-r)j}q^{(s-r){\binom j 2}} z^j \;,
\end{equation}
where
$$(a;q)_n=\prod_{j=0}^n (1-aq^j)\;. $$
The very-well-poised basic hypergeomeric series is defined as \cite[Ex.\ (2.1.11)]{gr}
\begin{equation}
\whyp {r+1}r{a_1, a_4, a_5, \ldots,a_{r+1};q,z} = \qhyp {r+1}r{a_1,qa_1^\frac12,-qa_1^\frac12,a_4,\ldots,a_{r+1}}{a_1^\frac12, -a_1^\frac12, q a_1/a_4,\ldots, qa_1/a_{r+1}}{q,z} \;.
\end{equation}

We will assume that $|q|<1$. We will also write
$$\theta(z;q)=(z,q/z;q)_\infty\;.$$
This theta function satisfies the quasi-periodicity
\begin{equation}\label{qp}\theta(zq^N;q)=\frac{(-1)^N}{q^{\binom N2}z^N}\theta(z;q)\;,\qquad N\in\mathbb Z\;.\end{equation}

We will formulate four fundamental  identities, which evaluate
a combination of a basic hypergeometric integral and sum. 
In each case, we assume that the parameters are generic, so that the poles of the integrand split naturally
into geometric sequences converging to $0$ and to $\infty$. The integration is over a positively oriented contour separating these two types of poles.

To prove the first identity, we use the Nasrallah--Rahman integral and the nonterminating Jackson summation, which are  top level results
for basic hypergeometric integral evaluations and summations,
respectively.
Consequently, we expect that Theorem \ref{awt}
is a top level result for evaluations of the type considered here, 
with combined integration and summation.

\begin{theorem}\label{awt}
Let $a_j$ be generic numbers and $N_j$ integers satisfying $a_1\dotsm a_6=q$
and $N_1+\dots+N_6=0$. Then,
\begin{multline}\label{awti}\sum_{m=-\infty}^\infty\oint \prod_{j=1}^6\frac{(q^{1+m/2}/a_jz,q^{1-m/2}z/a_j;q)_\infty}{(q^{N_j+m/2}a_jz,q^{N_j-m/2}a_j/z;q)_\infty}\frac{(1-q^mz^2)(1-q^mz^{-2})}{q^mz^{6m}}\frac{dz}{2\pi\ti z}\\
=\frac{2}{\prod_{j=1}^6q^{\binom {N_j}2}a_j^{N_j}}\prod_{1\leq j<k\leq 6}\frac{(q/a_ja_k;q)_\infty}{(a_ja_kq^{N_j+N_k};q)_\infty}\;. \end{multline}
\end{theorem}

\begin{proof}
Let $L$ denote the left-hand side of \eqref{awti}. 
Note that the poles of the integrand are situated at fixed values
of $zq^{m/2}$. Thus, we may replace
$z$ by $zq^{-m/2}$ and interchange the sum and the integral. This gives
\begin{align}\nonumber L&=\oint \prod_{j=1}^6\frac{(qz^\pm/a_j;q)_\infty}{(q^{N_j}a_jz^\pm;q)_\infty}{(1-z^2)(1-z^{-2})}\\
\label{l}&\quad\times{}_8\psi_8\left(\begin{matrix}q/z,-q/z,a_1/z,\dots,a_6/z\\ 1/z,-1/z,q/a_1z,\dots,q/a_6z\end{matrix};q,q\right)\frac{dz}{2\pi\ti z}\;.
\end{align}
By \cite[Eq.\ (III.38)]{gr}, we may write
\begin{multline}\label{pwt} \prod_{j=1}^6{(qz^\pm/a_j;q)_\infty}{(1-z^2)(1-z^{-2})}\,{}_8\psi_8\left(\begin{matrix}q/z,-q/z,a_1/z,\dots,a_6/z\\ 1/z,-1/z,q/a_1z,\dots,q/a_6z\end{matrix};q,q\right)\\
=\frac{(q;q)_\infty\prod_{j=1}^4(qa_5^{\pm}/a_j;q)_\infty\theta(a_6z^{\pm};q)(z^{\pm 2};q)_\infty}{(qa_5^2,a_6a_5^{\pm};q)_\infty}\\
\times{}_8W_7(a_5^2;a_5a_1,a_5a_2,a_5a_3,a_5a_4,a_5a_6;q,q)
+\operatorname{idem}(a_5;a_6)\;,\end{multline}
where the second term means that the first term is repeated with $a_5$ and $a_6$ interchanged. Using \eqref{qp} to write
$$\theta(a_6z^{\pm};q)_\infty=q^{2\binom{N_6}2}a_6^{2N_6}\theta(a_6q^{N_6}z^{\pm};q)_\infty\;, $$
this leads to
\begin{align*}
L&=q^{2\binom{N_6}2}a_6^{2N_6}\frac{(q;q)_\infty\prod_{j=1}^4(qa_5^{\pm}/a_j;q)_\infty}{(qa_5^2,a_6a_5^{\pm};q)_\infty}\,{}_8W_7(a_5^2;a_5a_1,a_5a_2,a_5a_3,a_5a_4,a_5a_6;q,q)\\
&\quad\times\oint  \frac{(z^{\pm 2},q^{1-N_6}a_6^{-1}z^\pm;q)_\infty}{\prod_{j=1}^5(q^{N_j}a_jz^\pm;q)_\infty}\frac{dz}{2\pi\ti z}+\operatorname{idem}\big((a_5,N_5);(a_6,N_6)\big)\;.
\end{align*}
Applying the Nasrallah--Rahman identity \cite[Eq.\ (6.4.1)]{gr}
$$\oint\frac{(z^{\pm 2}, B z^{\pm};q)_\infty}{\prod_{j=1}^5(b_jz^{\pm};q)_\infty}\frac{dz}{2\pi\ti z}= \frac{2\prod_{j=1}^5(B/b_j;q)_\infty}{(q;q)_\infty\prod_{1\leq j<k\leq 5}(b_jb_k;q)_\infty}\;,\qquad B=b_1\dotsm  b_5\;, $$
we conclude that
\begin{align*}
L&=2q^{2\binom{N_6}2}a_6^{2N_6}\frac{\prod_{j=1}^4(qa_5^{\pm}/a_j;q)_\infty\prod_{j=1}^5(q^{1-N_6-N_j}/a_ja_6;q)_\infty}{(qa_5^2,a_6a_5^{\pm};q)_\infty\prod_{1\leq j<k\leq 5}(q^{N_j+N_k}a_ja_k;q)_\infty}\\
&\quad\times \,{}_8W_7(a_5^2;a_5a_1,a_5a_2,a_5a_3,a_5a_4,a_5a_6;q,q)
+\operatorname{idem}((a_5,N_5);(a_6,N_6))\;.
\end{align*}
By the non-terminating Jackson summation \cite[Eq.\ (II.25)]{gr}, this can be simplified to the right-hand side of \eqref{awti}.
\end{proof}

If one formally replaces $6$ by $4$ in Theorem \ref{awt}, it is possible to replace the discrete parameters  $N_j$  by generic complex numbers. 
The proof of the corresponding identity is in fact very easy.

\begin{proposition}\label{wps}
For $a_j$ and $b_j$ generic,
\begin{multline}\label{wpsi}\sum_{m=-\infty}^\infty\oint \prod_{j=1}^4\frac{(q^{1+m/2}/a_jz,q^{1-m/2}z/a_j;q)_\infty}{(q^{m/2}b_jz,q^{-m/2}b_j/z;q)_\infty}\frac{(1-q^mz^2)(1-q^mz^{-2})}{q^mz^{4m}}\frac{dz}{2\pi\ti z}\\
=\frac{2(b_1b_2b_3b_4;q)_\infty}{(q/a_1a_2a_3a_4;q)_\infty}\prod_{1\leq j<k\leq 4}\frac{(q/a_ja_k;q)_\infty}{(b_jb_k;q)_\infty}\;. \end{multline}
\end{proposition}

\begin{proof}
With $L$ the left-hand side of \eqref{wpsi},
 the identity \eqref{l} is replaced by
\begin{align*} L&=\oint \prod_{j=1}^4\frac{(qz^\pm/a_j;q)_\infty}{(b_jz^\pm;q)_\infty}{(1-z^2)(1-z^{-2})}\\
&\quad\times{}_6\psi_6\left(\begin{matrix}q/z,-q/z,a_1/z,\dots,a_4/z\\ 1/z,-1/z,q/a_1z,\dots,q/a_4z\end{matrix};q,\frac q{a_1\dotsm a_4}\right)\frac{dz}{2\pi\ti z}\;.
\end{align*}
Applying Bailey's summation \cite[Eq.\ (II.33)]{gr} gives
$$L=\frac{(q;q)_\infty\prod_{1\leq j<k\leq 4}(q/a_ja_k;q)_\infty}{(q/a_1a_2a_3a_4;q)_\infty}
\oint\frac{(z^{\pm 2};q)_\infty}{\prod_{j=1}^4(b_jz^\pm;q)_\infty}\frac{dz}{2\pi\ti z}\;, $$
which reduces the result to the Askey--Wilson integral \cite[Eq.\ (6.1.4)]{gr}.
\end{proof}

The next result was obtained in \cite{Gahramanov:2013rda} for $M_j\equiv N_j\equiv 0$ and announced in \cite{Gahramanov:2014ona} in general.

\begin{theorem}\label{t}
Let $a_j$, $b_j$ be generic numbers and $M_j$, $N_j$ integers satisfying $a_1a_2 a_3=b_1b_2 b_3=q^{1/2}$
and $M_1+M_2+M_3=N_1+N_2+N_3=0$. Then,
\begin{multline}\label{p}\sum_{m=-\infty}^\infty\oint \prod_{j=1}^3\frac{(q^{1+m/2}/a_jz,q^{1-m/2}z/b_j;q)_\infty}{(q^{M_j+m/2}a_jz,q^{N_j-m/2}b_j/z;q)_\infty}\frac{(-1)^m}{z^{3m}}\frac{dz}{2\pi\ti z}\\
=\frac{1}{\prod_{j=1}^3q^ {\binom{M_j}2+\binom{N_j}2}a_j^{M_j}b_j^{N_j}}\prod_{j,k=1}^3\frac{(q/a_jb_k;q)_\infty}{(a_jb_kq^{M_j+N_k};q)_\infty}\;
. \end{multline}
\end{theorem}

\begin{proof}
This can be proved similarly as the special case treated in
\cite{Gahramanov:2013rda}, so we will be very brief. Shrinking the contour of integration to zero, we  pick up residues at the poles
$$z=q^{k-\frac m2+N_j}b_j\;,\qquad j=1,2,3,\quad k\geq\max(0,m-N_j)\;. $$
Working out the sum of residues explicitly,  the left-hand side
of \eqref{p} can be written
\begin{align}\nonumber L&=\frac{(-1)^{N_1}}{q^{\frac 32\, N_1^2}b_1^{3N_1}}\frac{(qb_1/b_2,qb_1/b_3;q)_\infty}{(q^{N_2-N_1}b_2/b_1,q^{N_3-N_1}b_3/b_1;q)_\infty}\prod_{j=1}^3\frac{(q/a_jb_1;q)_\infty}{(q^{N_1+M_j}a_jb_1;q)_\infty}\\
\nonumber &\quad\times {}_3\phi_2\left(\begin{matrix}q^{M_1+N_1}a_1b_1,q^{M_2+N_1}a_2b_1,q^{M_3+N_1}a_3b_1\\ q^{1+N_1-N_2}b_1/b_2,q^{1+N_1-N_3}b_1/b_3\end{matrix};q,q\right)
{}_3\phi_2\left(\begin{matrix}a_1b_1,a_2b_1,a_3b_1\\ qb_1/b_2,qb_1/b_3\end{matrix};q,q\right)\\
\label{lxy}&\quad+\operatorname{idem}\big((b_1,N_1);(b_2,N_2),(b_3,N_3)\big)\;.
 \end{align}

Let
$$x_1=b_1(qb_1/b_2,qb_1/b_3;q)_\infty\prod_{j=1}^3(a_jb_2,a_jb_3;q)_\infty\,{}_3\phi_2\left(\begin{matrix}a_1b_1,a_2b_1,a_3b_1\\ qb_1/b_2,qb_1/b_3\end{matrix};q,q\right)$$
and let $x_2$ and $x_3$ be defined by the same expression with $b_1$ interchanged by
$b_2$ and $b_3$, respectively. Then, by the nonterminating $q$-Saalsch\"utz summation \cite[Eq.\ (II.24)]{gr},
\begin{equation}\label{xd}x_2-x_1=b_2\theta(b_1/b_2;q)\prod_{j=1}^3\theta(a_jb_3;q)\;. 
\end{equation}
Let $\tilde x_j$ denote the result of replacing
 $a_j$ by $a_jq^{M_j}$ and $b_j$ by $b_jq^{N_j}$ in $x_j$. By \eqref{qp},
under the same change of variables, the right hand side of \eqref{xd}
is  divided by
$C=\prod_{j=1}^3q^{\binom{M_j}2+2\binom{N_j}2}a_j^{M_j}b_j^{2N_j}$. Thus, if we define
$y_j=C\tilde x_j$, then $y_2-y_1=x_2-x_1$. 
 By symmetry,  $y_j=x_j+D$, where $D$ is independent of $j$. It follows that
$$(x_3-x_2)x_1y_1+(x_1-x_3)x_2y_2+(x_2-x_1)x_3y_3=(x_2-x_1)(x_3-x_2)(x_3-x_1)\;. $$
After simplification, this  identity reduces to the desired result.
\end{proof}


We conclude with the following identity. Note that  the
parameter $t$ can be removed by scaling $z\mapsto tz$, but it seems useful to keep it.

 \begin{proposition}\label{ramp}
 For $a$, $b$, $c$, $d$ and $t$ generic parameters and integer $N$, such that
$|q^{\frac{N+1}2}a^{-1}|<|t|<|q^{\frac{N-1}2}b|$,
\begin{multline}\label{rp}\sum_{m=-\infty}^\infty\oint \frac{(q^{1+m/2}/az,q^{1-m/2}z/b;q)_\infty}{(q^{m/2}cz,q^{-m/2}d/z;q)_\infty}\,t^mz^{N-m}\frac{dz}{2\pi\ti z}\\
=t^N\frac{(q/ab,-q^{\frac{1+N}2}ct,-q^{\frac{1-N}2}d/t;q)_\infty}{(cd,-q^{\frac{1+N}2}/at,-q^{\frac{1-N}2}t/b;q)_\infty}\;.\end{multline}
 \end{proposition}

\begin{proof}
Replacing $z$ by $zq^{-m/2}$ and changing the order of summation, we
find that the left-hand side is given by
$$L=\oint \frac{(q/az,qz/b;q)_\infty}{(cz,d/z;q)_\infty}\,z^{N}\,{}_1\psi_1\left(\begin{matrix}b/z\\ q/az\end{matrix};q,-\frac{q^{\frac{1-N}2}t}{b}\right)\frac{dz}{2\pi\ti z}\;. $$
Applying Ramanujan's summation \cite[Eq.\ (II.29)]{gr} gives 
$$L=\frac{(q,q/ab;q)_\infty}{(-q^{\frac{1+N}2}/at,-q^{\frac{1-N}2}t/b;q)_\infty}\oint \frac{\theta(-q^{\frac{N+1}2}z/t;q)}{(cz,d/z;q)_\infty}\,z^{N}\,\frac{dz}{2\pi\ti z} $$
(the restriction on $t$ is needed here for convergence). 
It remains to prove that
\begin{equation}\label{soi}\oint \frac{\theta(-q^{\frac{N+1}2}z/t;q)}{(cz,d/z;q)_\infty}\,z^{N}\,\frac{dz}{2\pi\ti z}=t^N\frac{(-q^{\frac{1+N}2}ct,-q^{\frac{1-N}2}d/t;q)}{(q,cd;q)_\infty}\;.\end{equation}
To  this end, we expand the integral
 as  the sum of residues at the
points $z=q^kd$, $k\geq 0$. By
\cite[Eq.\ (4.10.8)]{gr}, under the additional assumption
  $|q^{\frac{1+N}2}t/d|<1$, the sum of residues
converges and can be computed by the
  $q$-binomial theorem \cite[Eq.\ (II.3)]{gr}.
Since the left-hand side of \eqref{soi} is  analytic in $t$ for
$t\neq 0$, the result holds also without the restriction  $|q^{\frac{1+N}2}t/d|<1$.
\end{proof}

 Using \eqref{absolute}, it is easy to see that the case $a=b=c=d$, $N=0$ is equivalent to the identity proved in Appendix A1 of
\cite{Krattenthaler:2011da}. It may be remarked that our proof of the general case is simpler.

\section{Conclusions}

Similarly to four-dimensional dualities \cite{Spiridonov:2009za,Spiridonov:2011hf}, equality of the superconformal indices for dual theories in three dimensions leads to  new non-trivial integral identities  \cite{Gahramanov:2013rda,Krattenthaler:2011da,Kapustin:2011jm}. We have presented four new identities for basic hypergeometric integrals. More concretely, we studied the generalized superconformal index of confining theories in three dimensions that has the form of a basic hypergeometric integral. This kind of result is important for better understanding  the structure of three-dimensional supersymmetric dualities. Most dualities discussed in the work are known in the literature, but the verification of these dualities using the superconformal index technique is new.

We also presented so-called pentagon identities. They are especially interesting from the geometrical point of view, which interprets the pentagon relation as the $3-2$ Pachner move in the context of the $3d-3d$ correspondence. This relates different  decompositions of a polyhedron with five ideal vertices into ideal tetrahedra.

It would be interesting to study more general $SU(N)$ gauge theories, other gauge groups and other confining theories.

\medskip

{\bf Acknowledgements}.
IG wishes to thank the Chalmers University of Technology for warm hospitality where this work started.  The research of IG is supported in part by the SFB 647 ``Raum-Zeit-Materie. Analytische und Geometrische Strukturen'', the Research Training Group GK 1504 ``Mass, Spectrum, Symmetry'' and the International Max Planck Research School for Geometric Analysis, Gravitation and String Theory. IG is partially supported by an ESF Short Visit Grant 6454 within the framework of the ``Interactions of Low-Dimensional Topology and Geometry with Mathematical Physics (ITGP)'' network. HR is supported by the Swedish Science Research Council. We are particularly grateful to Jonas Pollok for valuable comments on a preliminary version of the paper.

\appendix

\section{A short review of $3d$ ${\mathcal N}=2$ theories}

\label{rs}

The subject is very broad, and  we only discuss  basic facts needed to obtain our results in Section \ref{sis}. We refer the reader to \cite{Seiberg:1996nz, Aharony:1997bx, Intriligator:1996ex} for more details. 

\subsection{Conventions}

The Clifford algebra in $2+1$ dimensions with metric $g_{\mu \nu}$ is
\begin{align}
\{\gamma_\mu, \gamma_\nu \} \ & = \ 2 g_{\mu \nu}, \\
[\gamma_{\mu}, \gamma_{\nu} ] \ & = \ -2 \ti\epsilon^{\mu \nu \lambda} \gamma_\lambda.
\end{align}
As a convenient representation we choose $\gamma^\mu$ as 
\begin{equation}
(\gamma^1)^\alpha_\beta = \ti \sigma_2\;, \;\;\; (\gamma^2)^\alpha_\beta = \sigma_3\;, \;\;\; (\gamma^3)^\alpha_\beta = \sigma_1 \;,
\end{equation}
where $\alpha, \beta$ are spinor indices in the defining representation of $SL(2, \mathbb{R})$. Spinor indices are contracted, raised and lowered with the anti-symmetric matrix 
\begin{equation}
C_{\alpha \beta} \ = \ -C_{\beta \alpha} \ = \ C^{\beta \alpha} \ = \ 
  \left( {\begin{array}{cc}
   0 & -\ti \\       \ti & 0 \      \end{array} } \right).
\end{equation}

\subsection{$\mathcal N=2$ SUSY algebra}

Besides the ordinary generators of the Poincar\'e algebra, the three-dimensional $\mathcal N=2$ SUSY algebra (as for $\mathcal N=1$ SUSY in four dimensions) has four real supercharges. They can be combined into a complex supercharge and its Hermitian conjugate
\begin{equation}
Q_{\alpha} \;\; \text{and} \;\; \bar{Q}_\alpha \;,
\end{equation}
where  $\alpha$ is a spinor index which goes from $1$ to $2(=\mathcal N)$. The part of the $\mathcal N=2$ SUSY algebra involving the supercharges can be written \cite{Aharony:1997bx}
\begin{equation}
\left\{Q_\alpha,{Q}_\beta\right\} = \left\{\bar{Q}_\alpha,\bar{Q}_\beta\right\} =0,
\end{equation}
\begin{equation}
\left\{Q_\alpha,\bar{Q}_\beta\right\}=2 \gamma_{\alpha\beta}^i P_i+2\ti \epsilon_{\alpha\beta}Z,
\end{equation}
where the bosonic generator $P_\mu$ is the momentum generator and $Z$ is a central charge which can be thought of as the reduced component of four-dimensional momentum. The automorphism group of the algebra is $U(1)$ R-symmetry which rotates the supercharges
\begin{equation}
[R,Q_\alpha] \ = \ - Q_\alpha \;.
\end{equation}

Here we are interested in superconformal theories. In this case, we have two additional bosonic generators, special conformal transformations $K_\mu$ and dilatations $D$ and two fermionic generators, $S_\alpha$ and $\bar{S}_\alpha$. The $\mathcal N=2$ superconformal algebra in three dimensions takes the form of the  supergroup \cite{Jafferis:2010un}
\begin{equation}
SO(3,2)_\text{conf} \times SO(2)_R \subseteq OSp(2|4) \;.
\end{equation}
In Euclidean signature it is
\begin{equation}
SO(4,1)_\text{conf} \times SO(2)_R \subseteq OSp(2|2,2) \;.
\end{equation}
The first factor is the conformal group and the second one is the R-symmetry. Note that in the superconformal case the algebra has a distinguished R-symmetry. The important relation of the superconformal algebra for our purposes is
\begin{equation}
\{ \bar{Q}_\alpha, \bar{S}_\beta \} \ = \ M_{\mu \nu} [\gamma^\mu, \gamma^\nu]_{\alpha \beta} + 2 \varepsilon_{\alpha \beta} D  - 2 \varepsilon_{\alpha \beta} R \;.
\end{equation}
In particular, we will use the  commutation relation
\begin{equation}
\{ \bar{Q}_1, \bar{S}_1 \} \ = 2 \Delta -2 R - 2 j_3 \;.
\end{equation}

\subsection{Multiplets}

The supersymmetry representations of $3d$ ${\mathcal N}=2$ theories are closely related to the representations of $4d$ ${\mathcal N}=1$ theories and can be directly obtained from these by dimensional reduction.

To obtain irreducible representations one must impose constraints. In order to do so it is useful to define supercovariant derivatives:
\begin{align}
D_\alpha & = \frac{\partial}{\partial \theta_\alpha} - \ti(\gamma^\mu \bar{\theta})_\alpha \partial_\mu\;, \\
\bar{D}_\alpha & = \frac{\partial}{\partial \bar{\theta}_\alpha} - \ti(\gamma^\mu {\theta})_\alpha \partial_\mu\;.
\end{align}

The simplest type of superfield is a chiral multiplet $\Phi$. It satisfies the constraint
\begin{equation}
\bar{D}_\alpha \Phi \ = \ 0 \;.
\end{equation}
It can be expanded as 
\begin{equation}
\Phi \ = \ \phi(y) +\sqrt{2} \theta \psi (y) +\theta^2 F (y) \;,
\end{equation}
where $\phi$ is a complex scalar field, $\psi$ is a complex Dirac fermion, $F$ is an auxiliary complex scalar,
 $\theta$ is a Grassmann coordinate and $y^\mu = x^\mu +\ti\theta \sigma^\mu \bar{\theta}$.

The vector multiplet consists of a real scalar field $\sigma$, a vector field $A_\mu$, a complex Dirac fermion $\lambda$ and a real auxiliary scalar field $D$.
Its expansion in Wess-Zumino gauge is given by
\begin{equation}
V \ = \ -\theta \sigma^\mu \bar{\theta} A_\mu (x) -\theta \bar{\theta} \sigma + \ti \theta \theta \bar{\theta} \bar{\lambda} (x) -\ti \bar{\theta} \bar{\theta} \theta \lambda (x) +\frac12 \theta \theta \bar{\theta} \bar{\theta} D(x) \;,
\end{equation}
The appearance of a real scalar field $\sigma$ is due to the component of the four-dimensional vector field in the reduced direction.

\section{Details of Example 1}

All contributions to the superconformal indices in Example 1 are as follows:

\begin{itemize}

\item Contribution of the chiral multiplets
\begin{align} 
\text{ind}_\Phi = \left\{\begin{array}{lc} \left[ z a_i \frac{q^{\frac{|n_i+m|}{2}}}{1-q}-z^{-1} a_i^{-1} \frac{q^{1+\frac{|n_i+m|}{2}}}{1-q} \right] + \left[ z^{-1} a_i \frac{q^{\frac{|n_i-m|}{2}}}{1-q}-z a_i^{-1} \frac{q^{1+\frac{|n_j-m|}{2}}}{1-q} \right] & \text{: Theory A}, \\
\left[ a_i a_j \frac{q^{\frac{|n_i+n_j|}{2}}}{1-q}-a_i^{-1} a_j^{-1} \frac{q^{1+\frac{|n_i+n_j|}{2}}}{1-q} \right]  & \text{: Theory B}.
\end{array}\right. 
\end{align}

\item Contribution of the vector multiplet
\begin{align} 
\text{ind}_{gauge} = \left\{\begin{array}{lc} - q^{\frac12 |m_i|} z^2 - q^{\frac12 |m_i|} z^{-2} & \text{: Theory A}, \\
\text{no vector multiplet} & \text{: Theory B}.
\end{array}\right.
\end{align}

\item Other contributions
\begin{align} 
q_{0}(m) =  \left\{\begin{array}{lc} {-\frac{|n_j+m|}{2}-\frac{|n_j-m|}{2}}  & \text{: Theory A}, \\
 -\frac{|n_i+n_j|}{2} & \text{: Theory B}.
\end{array}\right.
\end{align}

\begin{align} 
e^{\ti b_0} =  \left\{\begin{array}{lc} z^{-\sum_{i=1}^6 (\frac{|n_i+m|}{2}-\frac{|n_i-m|}{2})}  & \text{: Theory A}, \\
0 & \text{: Theory B}.
\end{array}\right.
\end{align}

\end{itemize}

\end{document}